\documentclass[acmsmall,nonacm,natbib=false]
{acmart}

\usepackage{tikz}
\usepackage{svg}
\usepackage{backnaur}
\usepackage[T1]{fontenc}

\newtheorem{theorem}{Theorem}[section]
\newtheorem*{theorem*}{Theorem}

\newtheorem*{proposition*}{Proposition}

\newtheorem*{lemma*}{Lemma}

\newtheorem*{corollary*}{Corollary}
\newtheorem*{conjecture*}{Conjecture}

\newtheorem*{fact*}{Fact}

\newtheorem*{exercise*}{Exercise}

\newtheorem*{hypothesis*}{Hypothesis}

\theoremstyle{definition}
\newtheorem{definition}[theorem]{Definition}
\newtheorem{mytheorem}[theorem]{Theorem}
\newtheorem{mycorollary}[theorem]{Corollary}
\newtheorem{notation}[theorem]{Notation}

\newtheorem{baseexample}[theorem]{Example}
\newenvironment{example}{%
  \begin{quote}%
  \begin{baseexample}%
}{%
  \end{baseexample}%
  \end{quote}%
}

\newtheorem{exercise-easy}[theorem]{Exercise}
\newtheorem{exercise-med}[theorem]{Exercise}
\newtheorem{exercise-hard}[theorem]{Exercise$^\star$}

\newtheorem*{claim*}{Claim}

\newtheorem*{remark*}{Remark}

\newtheorem*{observation*}{Observation}

\newcommand{\Write}{\mathcal{W}}
\newcommand{\Merge}{\mathcal{M}}
\newcommand{\Query}{\mathcal{Q}}
\newcommand{\Values}{\mathbb{V}}
\newcommand{\States}{\mathbb{S}}
\newcommand{\Inputs}{\mathbb{I}}
\newcommand{\X}{\mathbb{X}}

\newcommand{\Evaluate}{\mathcal{E}}
\newcommand{\Inputset}{\mathcal{I}}
\newcommand{\Problem}{\mathcal{P}}
\newcommand{\CF}{\mathcal{F}} %

\newcommand{\LocalClause}{\mathbb{LC}}
\newcommand{\Clause}{\mathbb{C}}
\newcommand{\Input}{\Inputset}

\usepackage{ifthen}
\newboolean{showcomments}
\setboolean{showcomments}{true} %

\newcommand{\mynote}[3]{%
  \ifthenelse{\boolean{showcomments}}{%
   \fbox{\bfseries\sffamily\scriptsize#1}%
   {\small$\blacktriangleright$\textsf{\emph{\color{#3}{#2}}}$\blacktriangleleft$}}%
  {%
   \@bsphack
   \@esphack
}%
}

\AtBeginDocument{%
  }

\setcopyright{acmlicensed}
\copyrightyear{2025}
\acmYear{2025}
\acmDOI{XXXXXXX.XXXXXXX}
\acmConference[PODC '25]{ACM Symposium on Principles of Distributed Computing}{June 16-20, 2025}{Huatulco, Mexico}
\acmISBN{978-1-4503-XXXX-X/18/06}

\RequirePackage[
  datamodel=acmdatamodel,
  style=acmnumeric,
  ]{biblatex}

\bibliography{references}

\begin{document}

\title{A Preliminary Model of Coordination-free Consistency}

\author{Shulu Li}
\email{shulu\_li@outlook.com}
\orcid{0009-0001-7289-6577}
\affiliation{%
  \institution{University of California, Berkeley}
  \city{Berkeley}
  \state{California}
  \country{USA}
}

\author{Edward A. Lee}
\email{eal@berkeley.edu}
\orcid{0000-0002-5663-0584}
\affiliation{%
  \institution{University of California, Berkeley}
  \city{Berkeley}
  \state{California}
  \country{USA}
}

\begin{abstract}

Building consistent distributed systems has largely depended on complex coordination strategies that are not only tricky to implement, but also take a toll on performance as they require nodes to wait for coordination messages. 
In this paper, we explore the conditions under which no coordination is required to guarantee consistency. We present a simple and succinct theoretical model for distributed computation that separates coordination from computation. The main contribution of this work is mathematically defining concepts in distributed computing such as strong eventual consistency, consistency, consistent under partition, confluence, coordination-free, and monotonicity. 
Based on these definitions, we prove necessary and sufficient conditions for strong eventual consistency and give a proof of the CALM theorem from a distributed computation perspective. 
\end{abstract}

\ccsdesc[500]{Theory of computation~Models of computation}
\ccsdesc[500]{Theory of computation~Distributed computing models}

\keywords{Model of Computation, Distributed Computation, Consistency, Coordination-free, Monotonicity, CALM Theorem}

\maketitle

\section{Introduction}\label{sec:intro}

Consistency has long been one of the core design goals in distributed systems. However, distributed algorithms that guarantee consistency such as Paxos \cite{lamport_part-time_1998} and 2PC \cite{goos_notes_1978} are not only tricky to implement, but more importantly, have a performance impact because replicas have to wait for coordination. This seems unsolvable with the consistency-availability trade-off codified in the CAP theorem \cite{brewer2000cap}, but, for certain problems, consistency is achievable without coordination~\cite{helland_building_2009,Shapiro:11:CRDT}.

In this paper, we explore the theoretical boundaries in coordination-free distributed computation.

\subsection{Related Work}

Traditional theoretical frameworks for building consistent distributed systems, such as the CAP theorem \cite{brewer2000cap,Brewer:17:CAP}, the CAL theorem \cite{LeeEtAl:23:CAL_CPS,LeeEtAl:23:CAL_IC}, and linearizability \cite{herlihy_linearizability_1990}, have primarily focused on providing solutions for arbitrary, general-purpose problems in distributed computing. These approaches often require coordination mechanisms to enforce consistency, leading to inherent trade-offs between consistency and availability. In particular, the CAP theorem demonstrates that achieving strong consistency, high availability, and partition tolerance simultaneously is impossible in the presence of network partitions, forcing system designers and developers to make difficult choices based on the application requirements.

In contrast, Conflict-free Replicated Data Types (CRDTs) \cite{Shapiro:11:CRDT} introduce a new paradigm in distributed systems by enabling high availability while relaxing the consistency guarantees to allow for weaker forms of consistency. CRDTs provide a guarantee of strong eventual consistency: all replicas of the system will eventually converge to the same state. This property makes CRDTs particularly well-suited for applications where high availability is critical, such as collaborative text editing \cite{nicolaescu_near_2016}\cite{litt_peritext_2022}, as well as in local-first applications \cite{kleppmann_automerge_nodate}, where users may interact with local copies of data that later synchronize with other replicas.

Although CRDTs have been extensively studied and implemented, the majority of prior work has focused on identifying sufficient conditions for achieving CRDT properties, with less attention paid to understanding the necessary conditions that must be met for a system to guarantee strong eventual consistency. One of the key insights from earlier studies \cite{Shapiro:11:CRDT} is that CRDTs can be modeled as join-semilattices, but the theoretical framework for this approach has primarily been concerned with providing sufficient conditions, rather than exploring the underlying requirements for eventual consistency.

One recent advancement in this line of research is the work of Laddad et al.~\cite{laddad_keep_2022}, which draws connections between CRDTs and the CALM theorem. First conjectured in the context of database theory at PODS 2010 \cite{hellerstein_datalog_2010, hellerstein_declarative_2010}, the CALM theorem \cite{hellerstein_keeping_2020} asserts that a problem can have a coordination-free, consistent implementation if and only if the problem is monotonic. This result was initially presented as a conjecture and later formalized through the use of relational transducers \cite{ameloot_relational_2010, ameloot_weaker_2016}, primarily within the database context. While the CALM theorem has been a valuable tool for reasoning about the coordination requirements of distributed systems, it has largely remained confined to the database domain.

\subsection{Our Contributions}

In this work, we aim to formally discuss coordination-free consistency in a distributed computation context.
Our contributions can be summarized as follows:
\begin{itemize}
\itemsep=0ex
    \item We present a simple and succinct model for computation that separates computation from coordination and is especially suitable for analyzing consistency.
    \item We identify ACID 2.0~\cite{helland_building_2009} as necessary and sufficient condition for strong eventual consistency and give a proof.
    \item We define consistency as a partial order on problem outputs, modeling the problem of distributed computation as a function on partially ordered sets.
    \item We give a formal definition of coordination-freeness as confluent and consistent under partition. 
    \item We present an interpretation of the CALM theorem from a distributed computation perspective with our definition of consistency and coordination-freeness, and give a proof of the theorem.
\end{itemize}

\subsection{Technical Overview}

The remainder of the paper is organized as follows. We first present a simple and succinct theoretical model for distributed computation. This model describes distributed computation as the evaluation of clauses consisting of writes and merges. The main contribution of this model is that it separates the coordination layer from the computation layer: the coordination layer limits the execution traces that can be produced, and the computation layer executes the actual calculation. This trait makes it especially intuitive to argue about coordination-freeness and formally model the output of distributed computation.

Using this model of distributed computation, we first prove necessary and sufficient conditions for strong eventual consistency. Strong eventual consistency is formally defined as eventual state convergence; two replicas have strongly eventually consistent state if their state is identical after they have received the same inputs. The conclusion we reach in this section is that ACID 2.0~\cite{helland_building_2009} (associative, commutative, idempotent, and distributed) is necessary and sufficient for strong eventual consistency.

When writing applications, it is often the program output that we care about rather than the program state. In Section \ref{sec:consistency}, we formally define a problem as a function from sets of inputs to output values. Our main contribution in this section is formally defining consistency as a partial order on problem outputs; this is both general enough to apply to a variety of problems that allow for weaker consistency, and also specific enough to capture consistency.

Building on the definition of problems, we further define the implementation of problems. An implementation consists of two parts: a coordination function and an instance of an abstract data type. This aligns with our model of distributed computation that separates coordination and computation. 

One of the key contributions of this work is giving a formal definition of coordination-freeness. Previous work on the CALM theorem either has not given a formal definition of coordination-freeness \cite{hellerstein_keeping_2020} or gave the definition in a database context of transducer networks \cite{ameloot_relational_2010}. Our definition of coordination-freeness focuses on two qualities: confluence and consistency under partition. The definition and reasons for it are laid out in more detail in Section \ref{sec:coordination-free}.

Lamport's happens-before relation of events in a distributed system can be viewed as a partial order \cite{lamport_time_1978}. With consistency defined as a partial order on the problem output, distributed computation is in essence a function on posets, from the partial order of events to the partial order of outputs. We further define monotonicity and give a formal proof of the CALM theorem from a distributed computation perspective.

\section{Model of Distributed Computation}\label{sec:intro}

We define our model of computation in terms of a particular form of abstract data types.

\begin{definition}[Abstract Data Type]\label{def:adt}
An abstract data type in this paper is a tuple consisting of three functions and one initial state $(\Write,\Query, \Merge, s_0)$ , defined on three domains:  $\mathbb{S}$, $\Inputs$, and $\mathbb{V}$. $\mathbb{S}$ is the set of possible states, $\Inputs$ is an input alphabet (a set of values that that an individual input can have), and $\mathbb{V}$ is the set of possible outputs.
\begin{itemize}
\itemsep=0ex
    \item The initial state is $s_0 \in \mathbb{S}$.
    \item The write function has the form $\Write:\mathbb{S}\times\Inputs\rightarrow\mathbb{S}$ , which takes in an external input and modifies the internal state.
    \item The query function $\Query:\mathbb{S} \rightarrow\mathbb{V}$ reads the internal state $s$ and produces an output $v \in \mathbb{V}$ without changing the state.
    \item The merge function $\Merge:\mathbb{S}\times\mathbb{S}\rightarrow\mathbb{S}$ takes in two states and produces a new state.
\end{itemize}
\end{definition}

\begin{figure}[ht]
    \centering
    \resizebox{0.7\columnwidth}{!}{
        \tikzset{every picture/.style={line width=0.75pt}} %

\begin{tikzpicture}[x=0.75pt,y=0.75pt,yscale=-1,xscale=1]

\draw   (80,60) -- (140,60) -- (140,100) -- (80,100) -- cycle ;
\draw    (60,70) -- (78,70) ;
\draw [shift={(80,70)}, rotate = 180] [color={rgb, 255:red, 0; green, 0; blue, 0 }  ][line width=0.75]    (10.93,-3.29) .. controls (6.95,-1.4) and (3.31,-0.3) .. (0,0) .. controls (3.31,0.3) and (6.95,1.4) .. (10.93,3.29)   ;
\draw    (60,90) -- (78,90) ;
\draw [shift={(80,90)}, rotate = 180] [color={rgb, 255:red, 0; green, 0; blue, 0 }  ][line width=0.75]    (10.93,-3.29) .. controls (6.95,-1.4) and (3.31,-0.3) .. (0,0) .. controls (3.31,0.3) and (6.95,1.4) .. (10.93,3.29)   ;
\draw    (140,80) -- (158,80) ;
\draw [shift={(160,80)}, rotate = 180] [color={rgb, 255:red, 0; green, 0; blue, 0 }  ][line width=0.75]    (10.93,-3.29) .. controls (6.95,-1.4) and (3.31,-0.3) .. (0,0) .. controls (3.31,0.3) and (6.95,1.4) .. (10.93,3.29)   ;
\draw   (240,60) -- (300,60) -- (300,100) -- (240,100) -- cycle ;
\draw    (220,70) -- (238,70) ;
\draw [shift={(240,70)}, rotate = 180] [color={rgb, 255:red, 0; green, 0; blue, 0 }  ][line width=0.75]    (10.93,-3.29) .. controls (6.95,-1.4) and (3.31,-0.3) .. (0,0) .. controls (3.31,0.3) and (6.95,1.4) .. (10.93,3.29)   ;
\draw    (220,90) -- (238,90) ;
\draw [shift={(240,90)}, rotate = 180] [color={rgb, 255:red, 0; green, 0; blue, 0 }  ][line width=0.75]    (10.93,-3.29) .. controls (6.95,-1.4) and (3.31,-0.3) .. (0,0) .. controls (3.31,0.3) and (6.95,1.4) .. (10.93,3.29)   ;
\draw    (300,80) -- (318,80) ;
\draw [shift={(320,80)}, rotate = 180] [color={rgb, 255:red, 0; green, 0; blue, 0 }  ][line width=0.75]    (10.93,-3.29) .. controls (6.95,-1.4) and (3.31,-0.3) .. (0,0) .. controls (3.31,0.3) and (6.95,1.4) .. (10.93,3.29)   ;
\draw   (410,80) -- (450,80) -- (450,100) -- (410,100) -- cycle ;
\draw    (380,70) -- (468,70) ;
\draw [shift={(470,70)}, rotate = 180] [color={rgb, 255:red, 0; green, 0; blue, 0 }  ][line width=0.75]    (10.93,-3.29) .. controls (6.95,-1.4) and (3.31,-0.3) .. (0,0) .. controls (3.31,0.3) and (6.95,1.4) .. (10.93,3.29)   ;
\draw    (450,90) -- (468,90) ;
\draw [shift={(470,90)}, rotate = 180] [color={rgb, 255:red, 0; green, 0; blue, 0 }  ][line width=0.75]    (10.93,-3.29) .. controls (6.95,-1.4) and (3.31,-0.3) .. (0,0) .. controls (3.31,0.3) and (6.95,1.4) .. (10.93,3.29)   ;
\draw    (390,90) -- (408,90) ;
\draw [shift={(410,90)}, rotate = 180] [color={rgb, 255:red, 0; green, 0; blue, 0 }  ][line width=0.75]    (10.93,-3.29) .. controls (6.95,-1.4) and (3.31,-0.3) .. (0,0) .. controls (3.31,0.3) and (6.95,1.4) .. (10.93,3.29)   ;
\draw    (390,70) -- (390,90) ;

\draw (97,72.4) node [anchor=north west][inner sep=0.75pt]    {$\mathcal{{\displaystyle W}}$};
\draw (259,72.4) node [anchor=north west][inner sep=0.75pt]    {$\mathcal{M}$};
\draw (421,83.4) node [anchor=north west][inner sep=0.75pt]    {$\mathcal{Q}$};
\draw (47,62.4) node [anchor=north west][inner sep=0.75pt]    {$s$};
\draw (47,82.4) node [anchor=north west][inner sep=0.75pt]    {$i$};
\draw (207,62.4) node [anchor=north west][inner sep=0.75pt]    {$s$};
\draw (371,62.4) node [anchor=north west][inner sep=0.75pt]    {$s$};
\draw (161,72.4) node [anchor=north west][inner sep=0.75pt]    {$s'$};
\draw (207,82.4) node [anchor=north west][inner sep=0.75pt]    {$s'$};
\draw (321,72.4) node [anchor=north west][inner sep=0.75pt]    {$s''$};
\draw (477,63.4) node [anchor=north west][inner sep=0.75pt]    {$s$};
\draw (477,83.4) node [anchor=north west][inner sep=0.75pt]    {$v$};

\end{tikzpicture}
    }
    \caption{Rendering of $\Write$, $\Merge$, and $\Query$}
    \label{fig:replicated-object}
\end{figure}

Figure \ref{fig:replicated-object} is a rendering of $\Write$, $\Merge$, and $\Query$ that will be used later to illustrate an execution. Notice that both $\Write$ and $\Merge$ have two inputs and one output, while $\Query$ gives an output $v$ based on the state without modifying it. 

The separation between the write function and the query function is important. This allows us to discuss the program state and the program output separately. Sometimes two programs can have the same program state but different outputs, and this can determine whether they require coordination or not, as shown in Examples \ref{example:deadlock3} and \ref{example:garbage-collection}.

In a distributed computing setting, there are several replicas of the abstract data type in different locations, and the replicas can send messages to each other through some communication method.
We use the $\Merge$ function to model the communication between replicas.
To implement the merge function, there is no need to send the entire state over the network in order for the merge function to execute; one practical implementation is to only send the parts needed by the merge function, for example a message containing the latest changes. 

\begin{definition}[Object]\label{def:object}
An object $\mathcal{O}$ is a single instance of an abstract data type.  
\end{definition}

\begin{definition}[Replicated Object]\label{def:replicated-object}
Replicated objects $\mathcal{RO}$ are several distinct instances of an abstract data type, where the intent is that they deliver consistent responses to queries. Each instance of the abstract data type is called a replica. 
\end{definition}

Since a singular object has no need for communication, it doesn't need a merge function $\Merge$. 

\begin{notation}[Object]
An object $\mathcal{O}$ of an abstract data type is represented by the tuple $\mathcal{O} = (\Write,\Query, s_0)$. A replicated object $\mathcal{RO}$ of an abstract data type is represented by $\mathcal{RO} = (\Write,\Query, \Merge, s_0)$.
In both cases, the domains $(\mathbb{S}, \Inputs, \mathbb{V})$ are implied. 
\end{notation}

\begin{notation}[Write and Merge]\label{notation:write-and-merge} 
Since $\Write$ and $\Merge$ both take in two inputs and produce one output, we introduce a notational shorthand for $\Write$ and $\Merge$. Both operations are left-side binding.
\begin{itemize}
\itemsep=0ex
    \item $s'=\Write(s,i)$ can be written as $s'=s \Write i$.
    \item $s''=\Merge(s,s')$ can be written as $s''=s \Merge s'$.
\end{itemize}
\end{notation}

\begin{example}[Distributed Execution]\label{example:execution-dag} 
Let us now look at an example of distributed execution of two replicas of a replicated object. Both replicas do a write, a merge, another write, and another merge.

\begin{figure}[ht]
    \centering
    
    \resizebox{0.7\columnwidth}{!}{
        \tikzset{every picture/.style={line width=0.75pt}} %

\begin{tikzpicture}[x=0.75pt,y=0.75pt,yscale=-1,xscale=1]

\draw   (80,60) -- (140,60) -- (140,100) -- (80,100) -- cycle ;
\draw    (20,70) -- (78,70) ;
\draw [shift={(80,70)}, rotate = 180] [color={rgb, 255:red, 0; green, 0; blue, 0 }  ][line width=0.75]    (10.93,-3.29) .. controls (6.95,-1.4) and (3.31,-0.3) .. (0,0) .. controls (3.31,0.3) and (6.95,1.4) .. (10.93,3.29)   ;
\draw    (60,90) -- (78,90) ;
\draw [shift={(80,90)}, rotate = 180] [color={rgb, 255:red, 0; green, 0; blue, 0 }  ][line width=0.75]    (10.93,-3.29) .. controls (6.95,-1.4) and (3.31,-0.3) .. (0,0) .. controls (3.31,0.3) and (6.95,1.4) .. (10.93,3.29)   ;
\draw   (80,140) -- (140,140) -- (140,180) -- (80,180) -- cycle ;
\draw    (20,150) -- (78,150) ;
\draw [shift={(80,150)}, rotate = 180] [color={rgb, 255:red, 0; green, 0; blue, 0 }  ][line width=0.75]    (10.93,-3.29) .. controls (6.95,-1.4) and (3.31,-0.3) .. (0,0) .. controls (3.31,0.3) and (6.95,1.4) .. (10.93,3.29)   ;
\draw    (60,170) -- (78,170) ;
\draw [shift={(80,170)}, rotate = 180] [color={rgb, 255:red, 0; green, 0; blue, 0 }  ][line width=0.75]    (10.93,-3.29) .. controls (6.95,-1.4) and (3.31,-0.3) .. (0,0) .. controls (3.31,0.3) and (6.95,1.4) .. (10.93,3.29)   ;
\draw   (200,60) -- (260,60) -- (260,100) -- (200,100) -- cycle ;
\draw   (200,140) -- (260,140) -- (260,180) -- (200,180) -- cycle ;
\draw   (320,60) -- (380,60) -- (380,100) -- (320,100) -- cycle ;
\draw    (260,70) -- (318,70) ;
\draw [shift={(320,70)}, rotate = 180] [color={rgb, 255:red, 0; green, 0; blue, 0 }  ][line width=0.75]    (10.93,-3.29) .. controls (6.95,-1.4) and (3.31,-0.3) .. (0,0) .. controls (3.31,0.3) and (6.95,1.4) .. (10.93,3.29)   ;
\draw    (300,90) -- (318,90) ;
\draw [shift={(320,90)}, rotate = 180] [color={rgb, 255:red, 0; green, 0; blue, 0 }  ][line width=0.75]    (10.93,-3.29) .. controls (6.95,-1.4) and (3.31,-0.3) .. (0,0) .. controls (3.31,0.3) and (6.95,1.4) .. (10.93,3.29)   ;
\draw   (320,140) -- (380,140) -- (380,180) -- (320,180) -- cycle ;
\draw    (260,150) -- (318,150) ;
\draw [shift={(320,150)}, rotate = 180] [color={rgb, 255:red, 0; green, 0; blue, 0 }  ][line width=0.75]    (10.93,-3.29) .. controls (6.95,-1.4) and (3.31,-0.3) .. (0,0) .. controls (3.31,0.3) and (6.95,1.4) .. (10.93,3.29)   ;
\draw    (300,170) -- (318,170) ;
\draw [shift={(320,170)}, rotate = 180] [color={rgb, 255:red, 0; green, 0; blue, 0 }  ][line width=0.75]    (10.93,-3.29) .. controls (6.95,-1.4) and (3.31,-0.3) .. (0,0) .. controls (3.31,0.3) and (6.95,1.4) .. (10.93,3.29)   ;
\draw   (440,60) -- (500,60) -- (500,100) -- (440,100) -- cycle ;
\draw    (500,70) -- (558,70) ;
\draw [shift={(560,70)}, rotate = 180] [color={rgb, 255:red, 0; green, 0; blue, 0 }  ][line width=0.75]    (10.93,-3.29) .. controls (6.95,-1.4) and (3.31,-0.3) .. (0,0) .. controls (3.31,0.3) and (6.95,1.4) .. (10.93,3.29)   ;
\draw   (440,140) -- (500,140) -- (500,180) -- (440,180) -- cycle ;
\draw    (500,150) -- (558,150) ;
\draw [shift={(560,150)}, rotate = 180] [color={rgb, 255:red, 0; green, 0; blue, 0 }  ][line width=0.75]    (10.93,-3.29) .. controls (6.95,-1.4) and (3.31,-0.3) .. (0,0) .. controls (3.31,0.3) and (6.95,1.4) .. (10.93,3.29)   ;
\draw    (140,70) -- (198,70) ;
\draw [shift={(200,70)}, rotate = 180] [color={rgb, 255:red, 0; green, 0; blue, 0 }  ][line width=0.75]    (10.93,-3.29) .. controls (6.95,-1.4) and (3.31,-0.3) .. (0,0) .. controls (3.31,0.3) and (6.95,1.4) .. (10.93,3.29)   ;
\draw    (140,150) -- (198,150) ;
\draw [shift={(200,150)}, rotate = 180] [color={rgb, 255:red, 0; green, 0; blue, 0 }  ][line width=0.75]    (10.93,-3.29) .. controls (6.95,-1.4) and (3.31,-0.3) .. (0,0) .. controls (3.31,0.3) and (6.95,1.4) .. (10.93,3.29)   ;
\draw    (150,70) -- (199.11,168.21) ;
\draw [shift={(200,170)}, rotate = 243.43] [color={rgb, 255:red, 0; green, 0; blue, 0 }  ][line width=0.75]    (10.93,-3.29) .. controls (6.95,-1.4) and (3.31,-0.3) .. (0,0) .. controls (3.31,0.3) and (6.95,1.4) .. (10.93,3.29)   ;
\draw    (150,150) -- (198.72,91.54) ;
\draw [shift={(200,90)}, rotate = 129.81] [color={rgb, 255:red, 0; green, 0; blue, 0 }  ][line width=0.75]    (10.93,-3.29) .. controls (6.95,-1.4) and (3.31,-0.3) .. (0,0) .. controls (3.31,0.3) and (6.95,1.4) .. (10.93,3.29)   ;
\draw    (390,150) -- (438.72,91.54) ;
\draw [shift={(440,90)}, rotate = 129.81] [color={rgb, 255:red, 0; green, 0; blue, 0 }  ][line width=0.75]    (10.93,-3.29) .. controls (6.95,-1.4) and (3.31,-0.3) .. (0,0) .. controls (3.31,0.3) and (6.95,1.4) .. (10.93,3.29)   ;
\draw    (380,70) -- (438,70) ;
\draw [shift={(440,70)}, rotate = 180] [color={rgb, 255:red, 0; green, 0; blue, 0 }  ][line width=0.75]    (10.93,-3.29) .. controls (6.95,-1.4) and (3.31,-0.3) .. (0,0) .. controls (3.31,0.3) and (6.95,1.4) .. (10.93,3.29)   ;
\draw    (380,150) -- (438,150) ;
\draw [shift={(440,150)}, rotate = 180] [color={rgb, 255:red, 0; green, 0; blue, 0 }  ][line width=0.75]    (10.93,-3.29) .. controls (6.95,-1.4) and (3.31,-0.3) .. (0,0) .. controls (3.31,0.3) and (6.95,1.4) .. (10.93,3.29)   ;
\draw    (390,70) -- (439.11,168.21) ;
\draw [shift={(440,170)}, rotate = 243.43] [color={rgb, 255:red, 0; green, 0; blue, 0 }  ][line width=0.75]    (10.93,-3.29) .. controls (6.95,-1.4) and (3.31,-0.3) .. (0,0) .. controls (3.31,0.3) and (6.95,1.4) .. (10.93,3.29)   ;
\draw  [dash pattern={on 0.84pt off 2.51pt}] (30,62) .. controls (30,55.37) and (35.37,50) .. (42,50) -- (508,50) .. controls (514.63,50) and (520,55.37) .. (520,62) -- (520,98) .. controls (520,104.63) and (514.63,110) .. (508,110) -- (42,110) .. controls (35.37,110) and (30,104.63) .. (30,98) -- cycle ;
\draw  [dash pattern={on 0.84pt off 2.51pt}] (30,142) .. controls (30,135.37) and (35.37,130) .. (42,130) -- (508,130) .. controls (514.63,130) and (520,135.37) .. (520,142) -- (520,178) .. controls (520,184.63) and (514.63,190) .. (508,190) -- (42,190) .. controls (35.37,190) and (30,184.63) .. (30,178) -- cycle ;

\draw (3,62.4) node [anchor=north west][inner sep=0.75pt]    {$s_{0}$};
\draw (45,82.4) node [anchor=north west][inner sep=0.75pt]    {$i_{1}$};
\draw (3,142.4) node [anchor=north west][inner sep=0.75pt]    {$s_{0}$};
\draw (45,162.4) node [anchor=north west][inner sep=0.75pt]    {$i_{2}$};
\draw (97,72.4) node [anchor=north west][inner sep=0.75pt]    {$\mathcal{{\displaystyle W}}$};
\draw (97,152.4) node [anchor=north west][inner sep=0.75pt]    {$\mathcal{{\displaystyle W}}$};
\draw (221,72.4) node [anchor=north west][inner sep=0.75pt]    {$\mathcal{M}$};
\draw (221,152.4) node [anchor=north west][inner sep=0.75pt]    {$\mathcal{M}$};
\draw (337,72.4) node [anchor=north west][inner sep=0.75pt]    {$\mathcal{{\displaystyle W}}$};
\draw (337,152.4) node [anchor=north west][inner sep=0.75pt]    {$\mathcal{{\displaystyle W}}$};
\draw (461,72.4) node [anchor=north west][inner sep=0.75pt]    {$\mathcal{M}$};
\draw (461,152.4) node [anchor=north west][inner sep=0.75pt]    {$\mathcal{M}$};
\draw (141,50.4) node [anchor=north west][inner sep=0.75pt]    {$s_{1}$};
\draw (141,130.4) node [anchor=north west][inner sep=0.75pt]    {$s_{2}$};
\draw (261,50.4) node [anchor=north west][inner sep=0.75pt]    {$s_{3}$};
\draw (261,130.4) node [anchor=north west][inner sep=0.75pt]    {$s_{4}$};
\draw (381,50.4) node [anchor=north west][inner sep=0.75pt]    {$s_{5}$};
\draw (381,130.4) node [anchor=north west][inner sep=0.75pt]    {$s_{6}$};
\draw (563,60.4) node [anchor=north west][inner sep=0.75pt]    {$s_{7}$};
\draw (561,140.4) node [anchor=north west][inner sep=0.75pt]    {$s_{8}$};
\draw (285,82.4) node [anchor=north west][inner sep=0.75pt]    {$i_{3}$};
\draw (285,160.4) node [anchor=north west][inner sep=0.75pt]    {$i_{4}$};
\draw (31,32) node [anchor=north west][inner sep=0.75pt]   [align=left] {{\small Replica 1}};
\draw (31,192) node [anchor=north west][inner sep=0.75pt]   [align=left] {{\small Replica 2}};

\end{tikzpicture}
    }
    \caption{An Example of Distributed Execution}
    \label{fig:replicated-example}
\end{figure}
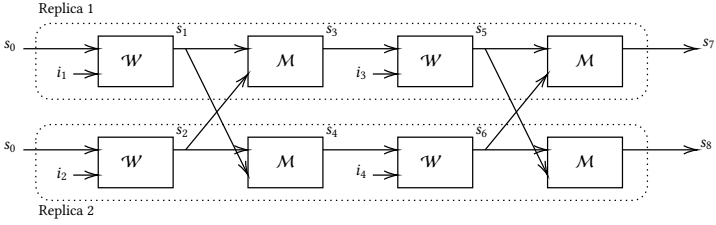

Figure \ref{fig:replicated-example} shows the diagram for a distributed execution. The top row represents the operations of replica one, and the bottom row represents the operations of replica two. Lines that cross the two rows mean network communication between the two replicas. Note that some inputs are made available to replica 1 and some to replica 2.

Let us take a closer look at the $s_3$ after the first merge in replica one.
To get to state $s_3$, replica 1 first applies input $i_1$ on the initial state $s_0$ with the $\Write$ operation to produce the state $s_1$. Then the state $s_1$ of replica 1 is merged with the state $s_2$ of replica 2 to produce state $s_3$. We can represent $s_3$ as $s_3=\Merge(s_1, s_2)=\Merge(\Write(s_0,i_1), \Write(s_0, i_2))$.
With our infix notation, this can be written $(s_0 \Write i_1) \Merge (s_0 \Write i_2)$.

Such expressions of states can always be reduced to refer to only the initial state $s_0$ and the inputs. We call such a reduced form a ``clause.'' For example, the clause that represents $s_3$ is $c_3=(s_0 \Write i_1) \Merge (s_0 \Write i_2)$. The difference between states and clauses is that states are values, and clauses are abstract formulas demonstrating how the state is produced. 

Similarly, we can represent $s_8$ as a clause.
To start, note that
$$s_8=\Merge(s_6,s_5)=\Merge(\Write(s_4,i_4), \Write(s_3, i_3)),$$ which can be reduced to
$$
c_8 = ((s_0 \Write i_2) \Merge (s_0 \Write i_1) \Write i_4) \Merge ((s_0 \Write i_1) \Merge (s_0 \Write i_2) \Write i_3).
$$
\end{example}

As shown in the previous example, we use clauses to model execution traces and states to denote the state of a replica after some execution. 
In general, a clause has the form 
$$\bnfpn{$c$} \bnfpo  \bnftd{$s_0$} \bnfor \bnfpn{$c$} \bnftd{$\Write$} \bnftd{$i$} \bnfor \bnfpn{$c$} \bnftd{$\Merge$} \bnfpn{$c$},$$
where $i\in \Inputs$ is an input.
The set of local clauses is the set of execution traces that could exist under local execution for an abstract data type (with no merge operations), and the set of all clauses is the set of traces that could exist under distributed execution.

\begin{definition}[Clause]\label{def:clause} The set of legal clauses $\Clause$ and the set of legal local clauses $\LocalClause$ are defined as follows in Backus–Naur form
 \cite{mccracken_backus-naur_2003}.
\begin{bnf*}
\bnfprod{$\Clause$} { \bnftd{$s_0$} \bnfor \bnfpn{$\Clause$} \bnfts{$\Write$} \bnfpn{Input} \bnfor \bnfpn{$\Clause$} \bnfts{$\Merge$} \bnfpn{$\Clause$}} \\
\bnfprod{$\LocalClause$} { \bnftd{$s_0$} \bnfor \bnfpn{$\LocalClause$} \bnfts{$\Write$} \bnfpn{Input}} \\
\bnfprod{Input} { \bnftd{$i\in\Inputs$} }
\end{bnf*}
\end{definition}

To transform clauses into states, we use an evaluation function $\Evaluate$.

\begin{definition}[Clause Evaluation]\label{def:clause-eval}
The evaluation function $\Evaluate_{\Write, \Merge}: \Clause \rightarrow \mathbb{S}$ calculates the resulting state of any valid formula by applying $\Write$ and $\Merge$ functions.
\end{definition}

The evaluation function $\Evaluate$ and the subscript $\Write$ and $\Merge$ may be omitted when there is no ambiguity.

\section{Strong Eventual Consistency}\label{sec:crdt}

With the computation model for distributed computing defined, we now use this model to prove the necessary and sufficient conditions for strong eventual consistency. 
First, note that each clause uses a particular set of inputs from the alphabet $\Inputs$.

\begin{definition}[Clause Input Set]\label{def:clause-input-set}
The input set of a clause $c$ is $\Inputset(c) = \{i \in \Inputs ~|~ i \text{ is used in } c \}$. 
\end{definition}
Note that a particular input value $i \in \Inputs$ may appear more than once in a clause $c$, but it will only appear once in the set $\Input(c)$.
In Example \ref{example:execution-dag}, $\Input(c_3)=\{i_1, i_2\}$ and $\Input(c_8)=\{i_1, i_2, i_3, i_4\}$, assuming that all these $i_k$ values are distinct.

We say that a replicated object is strongly eventually consistent if any two clauses that use the same input set yield the same state. 

\begin{definition}[Strong Eventual Consistency]\label{def:SEC} A replicated object $\mathcal{RO}=(\Write,\Merge,\Query,s_0)$ is strongly eventually consistent (SEC) iff,
$$
\forall c_1, c_2 \in \mathbb{C}, \Input(c_1)=\Input(c_2) \Rightarrow \Evaluate(c_1) = \Evaluate(c_2) 
$$
    
\end{definition}

This definition has some subtleties that require careful construction of the model. For example, if a clause uses a particular input $i \in \Inputs$ more than once, then for the replicated object to be SEC, the write and merge functions must be such that the result is the same as for clauses that use the input $i$ only once.
We will see that this requires careful definition of the input set $\Inputs$ in a problem-specific way.

Our definition of SEC is more formal, but consistent with those in the literature. For example, in Conflict-free Replicated Data Types~\cite{Shapiro:11:CRDT}, strong eventual consistency is defined as eventual delivery, strong convergence, and termination.
Eventual delivery corresponds to seeing the same inputs, termination corresponds to the clause being finite, and strong convergence corresponds to yielding the same final state.

\begin{mytheorem}[Necessary Conditions for Strong Eventual Consistency]\label{theorem:necessary-SEC}
If a replicated object $\mathcal{RO}=(\mathcal{W},\mathcal{M},\mathcal{Q},s_0)$ is strongly eventually consistent, then $\forall c, c_1, c_2 \in \Clause$, $\forall i \in \Inputs$, we have that
\begin{enumerate}
    \item $\Evaluate(c \Write i)=\Evaluate(c \Merge (s_0 \Write i))$,

    \item Merge is associative: $ \Evaluate((c \Merge c_1) \Merge c_2) = \Evaluate(c \Merge (c_1 \Merge c_2))$,

    \item Merge is commutative: $\Evaluate(c_1 \Merge c_2) =\Evaluate( c_2 \Merge c_1 )$, and

    \item Merge is idempotent: $\Evaluate(c_1 \Merge c_1) = \Evaluate(c_1) $.
\end{enumerate}
\end{mytheorem}
\begin{proof} For each of the equations, we let the clause on the left side of the equation be $c_l$ and the clause on the right side of the equation be $c_r$. In all four cases, $\Input(c_l) = \Input(c_r)$, and since the replicated object is strongly eventually consistent, we conclude that $\Evaluate(c_l) = \Evaluate(c_r)$.
This establishes the equality asserted in each of the four cases.
\end{proof}

Conditions (2) through (4) are illustrated in Figure~\ref{fig:ACI}.

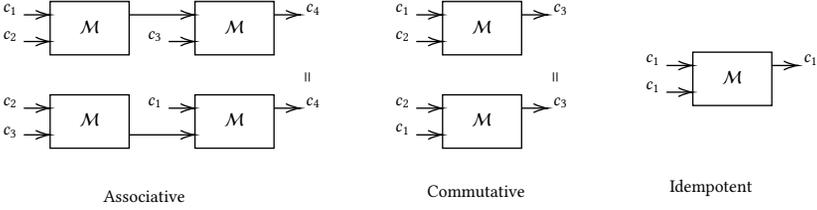
\begin{figure}[ht]
    \centering
    
    \resizebox{0.8\columnwidth}{!}{
        \tikzset{every picture/.style={line width=0.75pt}} %

\begin{tikzpicture}[x=0.75pt,y=0.75pt,yscale=-1,xscale=1]

\draw   (348,82) -- (408,82) -- (408,122) -- (348,122) -- cycle ;
\draw    (408,92) -- (426,92) ;
\draw [shift={(428,92)}, rotate = 180] [color={rgb, 255:red, 0; green, 0; blue, 0 }  ][line width=0.75]    (10.93,-3.29) .. controls (6.95,-1.4) and (3.31,-0.3) .. (0,0) .. controls (3.31,0.3) and (6.95,1.4) .. (10.93,3.29)   ;
\draw    (328,92) -- (346,92) ;
\draw [shift={(348,92)}, rotate = 180] [color={rgb, 255:red, 0; green, 0; blue, 0 }  ][line width=0.75]    (10.93,-3.29) .. controls (6.95,-1.4) and (3.31,-0.3) .. (0,0) .. controls (3.31,0.3) and (6.95,1.4) .. (10.93,3.29)   ;
\draw    (328,112) -- (346,112) ;
\draw [shift={(348,112)}, rotate = 180] [color={rgb, 255:red, 0; green, 0; blue, 0 }  ][line width=0.75]    (10.93,-3.29) .. controls (6.95,-1.4) and (3.31,-0.3) .. (0,0) .. controls (3.31,0.3) and (6.95,1.4) .. (10.93,3.29)   ;
\draw   (348,152) -- (408,152) -- (408,192) -- (348,192) -- cycle ;
\draw    (408,162) -- (426,162) ;
\draw [shift={(428,162)}, rotate = 180] [color={rgb, 255:red, 0; green, 0; blue, 0 }  ][line width=0.75]    (10.93,-3.29) .. controls (6.95,-1.4) and (3.31,-0.3) .. (0,0) .. controls (3.31,0.3) and (6.95,1.4) .. (10.93,3.29)   ;
\draw    (328,162) -- (346,162) ;
\draw [shift={(348,162)}, rotate = 180] [color={rgb, 255:red, 0; green, 0; blue, 0 }  ][line width=0.75]    (10.93,-3.29) .. controls (6.95,-1.4) and (3.31,-0.3) .. (0,0) .. controls (3.31,0.3) and (6.95,1.4) .. (10.93,3.29)   ;
\draw    (328,182) -- (346,182) ;
\draw [shift={(348,182)}, rotate = 180] [color={rgb, 255:red, 0; green, 0; blue, 0 }  ][line width=0.75]    (10.93,-3.29) .. controls (6.95,-1.4) and (3.31,-0.3) .. (0,0) .. controls (3.31,0.3) and (6.95,1.4) .. (10.93,3.29)   ;
\draw   (50,82) -- (110,82) -- (110,122) -- (50,122) -- cycle ;
\draw    (30,92) -- (48,92) ;
\draw [shift={(50,92)}, rotate = 180] [color={rgb, 255:red, 0; green, 0; blue, 0 }  ][line width=0.75]    (10.93,-3.29) .. controls (6.95,-1.4) and (3.31,-0.3) .. (0,0) .. controls (3.31,0.3) and (6.95,1.4) .. (10.93,3.29)   ;
\draw    (30,112) -- (48,112) ;
\draw [shift={(50,112)}, rotate = 180] [color={rgb, 255:red, 0; green, 0; blue, 0 }  ][line width=0.75]    (10.93,-3.29) .. controls (6.95,-1.4) and (3.31,-0.3) .. (0,0) .. controls (3.31,0.3) and (6.95,1.4) .. (10.93,3.29)   ;
\draw   (160,82) -- (220,82) -- (220,122) -- (160,122) -- cycle ;
\draw    (220,92) -- (238,92) ;
\draw [shift={(240,92)}, rotate = 180] [color={rgb, 255:red, 0; green, 0; blue, 0 }  ][line width=0.75]    (10.93,-3.29) .. controls (6.95,-1.4) and (3.31,-0.3) .. (0,0) .. controls (3.31,0.3) and (6.95,1.4) .. (10.93,3.29)   ;
\draw    (110,92) -- (158,92) ;
\draw [shift={(160,92)}, rotate = 180] [color={rgb, 255:red, 0; green, 0; blue, 0 }  ][line width=0.75]    (10.93,-3.29) .. controls (6.95,-1.4) and (3.31,-0.3) .. (0,0) .. controls (3.31,0.3) and (6.95,1.4) .. (10.93,3.29)   ;
\draw    (140,112) -- (158,112) ;
\draw [shift={(160,112)}, rotate = 180] [color={rgb, 255:red, 0; green, 0; blue, 0 }  ][line width=0.75]    (10.93,-3.29) .. controls (6.95,-1.4) and (3.31,-0.3) .. (0,0) .. controls (3.31,0.3) and (6.95,1.4) .. (10.93,3.29)   ;
\draw   (50,152) -- (110,152) -- (110,192) -- (50,192) -- cycle ;
\draw    (30,162) -- (48,162) ;
\draw [shift={(50,162)}, rotate = 180] [color={rgb, 255:red, 0; green, 0; blue, 0 }  ][line width=0.75]    (10.93,-3.29) .. controls (6.95,-1.4) and (3.31,-0.3) .. (0,0) .. controls (3.31,0.3) and (6.95,1.4) .. (10.93,3.29)   ;
\draw    (30,182) -- (48,182) ;
\draw [shift={(50,182)}, rotate = 180] [color={rgb, 255:red, 0; green, 0; blue, 0 }  ][line width=0.75]    (10.93,-3.29) .. controls (6.95,-1.4) and (3.31,-0.3) .. (0,0) .. controls (3.31,0.3) and (6.95,1.4) .. (10.93,3.29)   ;
\draw   (160,152) -- (220,152) -- (220,192) -- (160,192) -- cycle ;
\draw    (220,162) -- (238,162) ;
\draw [shift={(240,162)}, rotate = 180] [color={rgb, 255:red, 0; green, 0; blue, 0 }  ][line width=0.75]    (10.93,-3.29) .. controls (6.95,-1.4) and (3.31,-0.3) .. (0,0) .. controls (3.31,0.3) and (6.95,1.4) .. (10.93,3.29)   ;
\draw    (110,182) -- (158,182) ;
\draw [shift={(160,182)}, rotate = 180] [color={rgb, 255:red, 0; green, 0; blue, 0 }  ][line width=0.75]    (10.93,-3.29) .. controls (6.95,-1.4) and (3.31,-0.3) .. (0,0) .. controls (3.31,0.3) and (6.95,1.4) .. (10.93,3.29)   ;
\draw    (140,162) -- (158,162) ;
\draw [shift={(160,162)}, rotate = 180] [color={rgb, 255:red, 0; green, 0; blue, 0 }  ][line width=0.75]    (10.93,-3.29) .. controls (6.95,-1.4) and (3.31,-0.3) .. (0,0) .. controls (3.31,0.3) and (6.95,1.4) .. (10.93,3.29)   ;
\draw   (538,120) -- (598,120) -- (598,160) -- (538,160) -- cycle ;
\draw    (598,130) -- (616,130) ;
\draw [shift={(618,130)}, rotate = 180] [color={rgb, 255:red, 0; green, 0; blue, 0 }  ][line width=0.75]    (10.93,-3.29) .. controls (6.95,-1.4) and (3.31,-0.3) .. (0,0) .. controls (3.31,0.3) and (6.95,1.4) .. (10.93,3.29)   ;
\draw    (518,130) -- (536,130) ;
\draw [shift={(538,130)}, rotate = 180] [color={rgb, 255:red, 0; green, 0; blue, 0 }  ][line width=0.75]    (10.93,-3.29) .. controls (6.95,-1.4) and (3.31,-0.3) .. (0,0) .. controls (3.31,0.3) and (6.95,1.4) .. (10.93,3.29)   ;
\draw    (518,150) -- (536,150) ;
\draw [shift={(538,150)}, rotate = 180] [color={rgb, 255:red, 0; green, 0; blue, 0 }  ][line width=0.75]    (10.93,-3.29) .. controls (6.95,-1.4) and (3.31,-0.3) .. (0,0) .. controls (3.31,0.3) and (6.95,1.4) .. (10.93,3.29)   ;

\draw (369,94.4) node [anchor=north west][inner sep=0.75pt]    {$\mathcal{M}$};
\draw (311,82.4) node [anchor=north west][inner sep=0.75pt]    {$c_{1}$};
\draw (311,102.4) node [anchor=north west][inner sep=0.75pt]    {$c_{2}$};
\draw (431,82.4) node [anchor=north west][inner sep=0.75pt]    {$c_{3}$};
\draw (369,164.4) node [anchor=north west][inner sep=0.75pt]    {$\mathcal{M}$};
\draw (311,152.4) node [anchor=north west][inner sep=0.75pt]    {$c_{2}$};
\draw (311,172.4) node [anchor=north west][inner sep=0.75pt]    {$c_{1}$};
\draw (431,152.4) node [anchor=north west][inner sep=0.75pt]    {$c_{3}$};
\draw (71,94.4) node [anchor=north west][inner sep=0.75pt]    {$\mathcal{M}$};
\draw (13,82.4) node [anchor=north west][inner sep=0.75pt]    {$c_{1}$};
\draw (13,102.4) node [anchor=north west][inner sep=0.75pt]    {$c_{2}$};
\draw (181,94.4) node [anchor=north west][inner sep=0.75pt]    {$\mathcal{M}$};
\draw (123,102.4) node [anchor=north west][inner sep=0.75pt]    {$c_{3}$};
\draw (243,82.4) node [anchor=north west][inner sep=0.75pt]    {$c_{4}$};
\draw (71,164.4) node [anchor=north west][inner sep=0.75pt]    {$\mathcal{M}$};
\draw (123,152.4) node [anchor=north west][inner sep=0.75pt]    {$c_{1}$};
\draw (13,152.4) node [anchor=north west][inner sep=0.75pt]    {$c_{2}$};
\draw (181,164.4) node [anchor=north west][inner sep=0.75pt]    {$\mathcal{M}$};
\draw (13,174.4) node [anchor=north west][inner sep=0.75pt]    {$c_{3}$};
\draw (243,152.4) node [anchor=north west][inner sep=0.75pt]    {$c_{4}$};
\draw (559,132.4) node [anchor=north west][inner sep=0.75pt]    {$\mathcal{M}$};
\draw (501,120.4) node [anchor=north west][inner sep=0.75pt]    {$c_{1}$};
\draw (501,140.4) node [anchor=north west][inner sep=0.75pt]    {$c_{1}$};
\draw (621,120.4) node [anchor=north west][inner sep=0.75pt]    {$c_{1}$};
\draw (430.4,144.05) node [anchor=north west][inner sep=0.75pt]  [rotate=-269.85]  {${\displaystyle =}$};
\draw (242.4,144.05) node [anchor=north west][inner sep=0.75pt]  [rotate=-269.85]  {${\displaystyle =}$};
\draw (89,222) node [anchor=north west][inner sep=0.75pt]   [align=left] {Associative};
\draw (335,218) node [anchor=north west][inner sep=0.75pt]   [align=left] {Commutative};
\draw (519,214) node [anchor=north west][inner sep=0.75pt]   [align=left] {Idempotent};

\end{tikzpicture}
    }
    \caption{A Diagram Showing the ACI properties of $\Merge$}
    \label{fig:ACI}
\end{figure}

\begin{mytheorem}[Sufficient Conditions for Strong Eventual Consistency]\label{theorem:sufficient-SEC}
If a replicated object $\mathcal{RO}=(\mathcal{W},\mathcal{M},\mathcal{Q},s_0)$ satisfies $\forall c, c_1, c_2 \in \Clause$, $\forall i \in \Inputs$:

\begin{enumerate}
    \item $\Evaluate(c \Write i)= \Evaluate(c \Merge ( s_0 \Write i))$
    \item Merge is associative: $ \Evaluate((c \Merge c_1) \Merge c_2) = \Evaluate(c \Merge (c_1 \Merge c_2))$
    \item Merge is commutative: $\Evaluate(c_1 \Merge c_2) = \Evaluate(c_2 \Merge c_1 )$
    \item Merge is idempotent: $\Evaluate(c_1 \Merge c_1)  = \Evaluate(c_1)$
\end{enumerate}
Then the replicated object $\mathcal{RO}$ is strongly eventually consistent.
\end{mytheorem}

\begin{proof}
We prove this by transforming any two clauses with the same clause input set into the same form, and thus the evaluations of the two clauses are equal.

(Step 1) First we transform all $\Write$ in $c_1$ and $c_2$ using property (1), which implies that $s \Write i=s \Merge(s_0 \Write i)$. We denote $s_i = s_0 \Write i$. Now $c_1$ and $c_2$ consist only of $\Merge$ operations on $s_0, s_{i_1}, s_{i_2}, s_{i_3} ...$

(Step 2) Because the $\Merge$ operation is ACI (associative, commutative and idempotent), we can transform both $c_1$ and $c_2$ into the form of $s_0 \Merge s_{i_1} \Merge s_{i_2} \Merge s_{i_3} ...$. If $\Input(c_1) = \Input(c_2)$, then the transformed $c_1$ and $c_2$ can be made identical using property (3) by just putting the inputs in the same order.
Consequently,  $\Input(c_1) = \Input(c_2)$ implies that $\Evaluate(c_1)=\Evaluate(c_2)$, establishing that the replicated object is SEC.
\end{proof}

\begin{mycorollary}[Necessary and Sufficient Conditions for Strong Eventual Consistency]\label{corollary:necessary-sufficient-SEC}
A replicated object $\mathcal{RO}=(\mathcal{W},\mathcal{M},\mathcal{Q},s_0)$ 
is strongly eventually consistent 
$\Leftrightarrow$ 
$\forall c \in \Clause$, $\forall i \in \Inputs$, $\Evaluate(c \Write i)= \Evaluate(c \Merge ( s_0 \Write i))$, and $\Merge$ is associative, commutative, and idempotent.

\begin{proof}
This follows directly from Theorems \ref{theorem:necessary-SEC} and \ref{theorem:sufficient-SEC}.
\end{proof}

\end{mycorollary}

\section{Consistency}\label{sec:consistency}

In the previous sections, we only discussed the state of a distributed program. But in most applications, what we care about is the program output. Specifically, does the distributed program implement the problem that we are aiming to solve? To answer this question, we first have to formally define the problem that we are trying to solve.

A ``problem'' defines a result that should be given by an implementation of the problem after being provided with some inputs.
We distinguish an individual input value from the total input, a collection of input values.
Each individual input value may be part of a sequence, in which case the input values have a semantic ordering,
and the total input is a sequence.
Or, for other problems, the total input may be a collection of unordered input values.
The total input could even be a mix, such as a collection of sequences.

In all of these cases, each individual input will be drawn from an input alphabet $\Inputs$,
and a totality of inputs provided to the problem will be a subset $x \subseteq \Inputs$.
For some problems, not every subset of $\Inputs$ is a legal total input.
Hence, a part of the definition of a problem is a set $\X \subseteq 2^\Inputs$,
a set of legal subsets of $\Inputs$ that are legal total inputs.
For example, if inputs are semantically ordered, then a total input may be represented by a function
of the form $f\colon D_n \to Y$, where $D_n = \{0, 1, \cdots n \}$ for some $n \in \mathbb{N}$ and some set ${Y}$.
In this case, the input alphabet is:
$$
\Inputs = \mathbb{N} \times Y,
$$
and the set of legal total inputs is
$$
\X = \{ \text{graph}(f) ~|~ f\colon D_n \to Y, n \in \mathbb{N} \}
$$
where the graph of a function $f$ is the set of ordered pairs $(x,y)$ that $f(x)=y$.
So, for example, if $Y = \mathbb{N}$, then a legal total input is
$$
x = \{(0,0), (1,2), (2,4), (3,6)\} \in \X .
$$
Illegal input examples are $x = \{(0,0), (0,1)\}$ and $x = \{(0, 0), (2,4)\}$ because these are not graphs of function of the form $f\colon D_n \to Y$. In summary, we have the following definition:

\begin{definition}[Total Input]
    Given a set $\Inputs$, the input alphabet, a \emph{total input} is a subset $x \subseteq \Inputs$. The set of legal total inputs for a problem is $\X \subseteq 2^\Inputs$, a set of subsets of $\Inputs$.
\end{definition}

\begin{definition}[Problem]\label{def:problem}
A problem is a tuple $(\Problem, \X, \Values, \leq)$, where $\Problem$ is a function 
$$
\Problem:\X \rightarrow \Values,
$$
$\X \subseteq 2^\Inputs$ is the set of legal total inputs, $\Values$ is the set of output values, and $\leq$ is a partial-order relation on $\Values$.
\end{definition}

\begin{example}[Distributed Deadlock Detection]\label{example:deadlock} 

Here we look at the example of distributed deadlock detection from \cite{hellerstein_keeping_2020} and formally formulate it as a problem. 
The problem of distributed deadlock detection is essentially the problem of loop detection in a graph. Each node represents a thread waiting for a mutual exclusion lock, and each directed edge represents a dependency, where the start of the edge is a thread that must release the lock before the end of the edge can acquire it. Our goal is to identify loops in this graph.

\begin{figure}[ht]
    \centering
    \includegraphics[width=0.7\columnwidth]{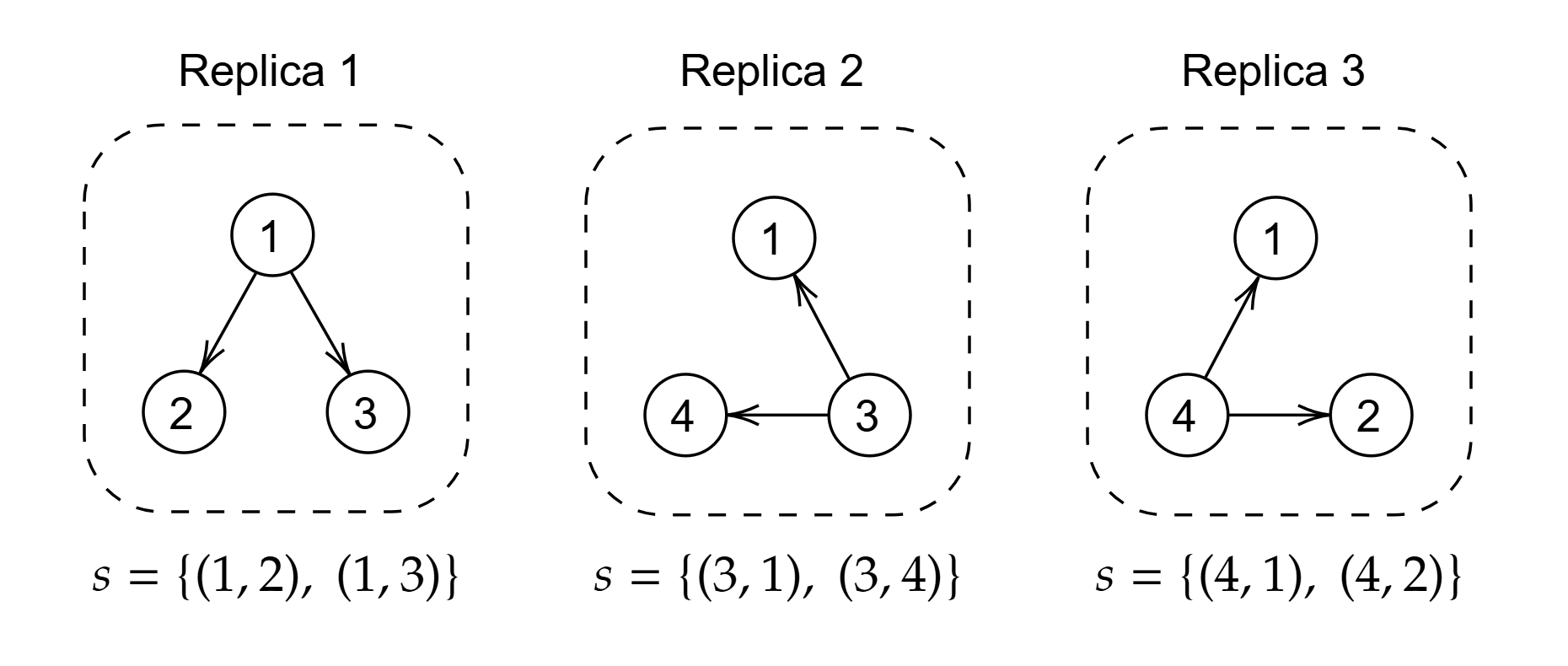}
    \caption{Distributed Deadlock Detection Example}
    \label{fig:distributed-deadlock-detection}
\end{figure}

The input alphabet $\Inputs$ is the set of all edges: $\Inputs = \{(i,j) | i,j \in \mathbb{N}\}  = \mathbb{N} \times \mathbb{N}$. An edge from node $i$ to node $j$ is $e = (i, j)$. 
The problem's domain in this particular definition is $\X = 2^\Inputs$. The output of the problem is the set of all edges that are in a loop, so the co-domain of the $\Problem$ function is the same, $\mathbb{V} = 2^\Inputs$, the set of all subsets of $\mathbb{N} \times \mathbb{N}$. For example, as shown in Figure~\ref{fig:distributed-deadlock-detection}, if a total input is $x = \{i_1=(1,2), i_2=(1,3), i_3=(3,1), i_4=(3,4), i_5=(4,1), i_6=(4,2)\}$, then the problem output is $$\Problem(\{i_1,i_2,i_3,i_4,i_5,i_6\}) = \{ (1,3), (3,1), (3,4), (4,1) \}.$$
A natural choice for the partial order relation $\leq$ on the set of output values $\Values$, the role of which we discuss next, is the subset relation $\subseteq$.
\end{example}

With the definition of problems, we can now define consistency in the context of a problem's outputs. For a problem with codomain $\mathbb{V}$, consistency is defined as a non-contradiction relation on $\mathbb{V}$. Intuitively, for any two problem outputs, if a problem output does not contradict another problem output, we say that the latter problem output is consistent with the first output.

\begin{definition}[Non-contradiction and Consistency]
Given a partial order $\leq$ defined on the set of outputs $\mathbb{V}$ for problem $(\Problem, \X, \Values, \leq)$, and given
$v_1, v_2 \in \mathbb{V} $, then we say $v_2$ is consistent with $v_1$ if $v_1 \leq v_2$. Alternatively, we say that $v_2$ does not contradict $v_1$.
\end{definition}

Consistency has to be a partial order; i.e., it satisfies reflexivity ($v_1 \leq v_1$), anti-symmetry ($v_1 \leq v_2 \land v_2 \leq v_1 \Rightarrow v_1=v_2$), and transitivity ($v_1 \leq v_2 \land v_2 \leq v_3 \Rightarrow v_1 \leq v_3$). Reflexivity is a natural result of consistency because an output never contradicts itself. Transitivity is also natural for consistency because if $v_1$ is consistent with $v_2$ and $v_2$ is consistent with $v_3$, it's natural that $v_1$ is consistent with $v_3$. We will discuss anti-symmetry in definition \ref{def:consistent-each-other}.

\begin{definition}[Strong Consistency] \label{def:consistent-each-other}
$\forall v_1, v_2 \in \mathbb{V}$, we say $v_1$ and $v_2$ are strongly consistent if they are consistent with each other, i.e. $v_1 \leq v_2$ and $v_2 \leq v_1$. Because $\leq$ is a partial order, this implies that $v_1=v_2$ because of anti-symmetry.
\end{definition}

\begin{example}[Distributed Deadlock Detection]\label{example:deadlock2} 

In the problem of deadlock detection in Example \ref{example:deadlock}, the consistency partial order is the subset order on $\mathbb{V}$. An output $v_2$ does not contradict a previous output $v_1$ if $v_2$ contains all the edges that $v_1$ contains.

For this example, consistency has the practical implication of ``safe to report deadlock.'' If all future outputs, given more inputs, are guaranteed to be consistent with the current output, then we can safely act on the current output. For example, the system can report to the user immediately that a deadlock has been detected, with the guarantee that this deadlock will not be resolved later.
\end{example}

\section{Coordination-free}\label{sec:coordination-free}

With the definition of problems, we now look into the implementations of problems. An implementation consists of two parts: a coordination function $\CF$ and a replicated object $\mathcal{RO}$. The two parts correspond to the two basic components in building distributed programs: the coordination layer and the computation layer.

\begin{definition}[Coordination Function]
A coordination function $\CF:\X \rightarrow 2^\Clause$ is a function that, given any total input, produces a set of clauses.
\end{definition}

In essence, a coordination function describes all possible execution traces given a total set of inputs. This formally defines the behavior of the coordination layer by limiting the clauses that could be produced. For example, an extreme case of a coordination function is one that never produces a clause with a merge operation $\Merge$. This coordination function represents a coordination layer that would only allow the implementation to run on a single machine.

\begin{definition}[Implementation]
An implementation of a problem $(\Problem, \X, \Values, \leq)$ consists of a coordination function $\CF:\X \rightarrow 2^\Clause$ and a replicated object $\mathcal{RO}=(\Write,\Query, \Merge, s_0)$ such that 
$$
\forall x \in \X, \forall c \in \CF(x), \Problem(x) = \Query(\Evaluate_{\Write, \Merge}(c)) 
$$
\end{definition}

This says that, for an implementation to actually implement a problem, any clause permitted by the coordination function must, when evaluated by the replicated object, yield the problem's answer.

\begin{example}[Distributed Deadlock Detection]\label{example:deadlock3}
  Here we present an implementation for the problem of distributed deadlock detection in Example \ref{example:deadlock}. To formally describe an implementation, we describe a coordination function $\CF$ and a replicated object ${\mathcal{RO}=(\Write,\Merge,\Query,s_0)}$. 

\begin{itemize}
\itemsep=0ex
    \item $\CF: \forall x \in \X, \CF(x)=\{c \in \Clause \mid \Inputset(c) = x\}$
    \item $s_0=\emptyset$
    \item $\Write(s, i)= s \cup {i}$
    \item $\Merge(s_1,s_2) = s_1 \cup s_2$
    \item $\Query(s) = \{ e_0 \in s ~|~ \exists e_1, e_2,..., e_n \in s, \forall i \in 0,1,...,n-1, v(e_i) = u(e_{i+1}) \mbox{ and } v(e_n) = u(e_0) \}$ where $u(e)$ and $v(e)$ are the source and destination nodes of the edge $e$.
\end{itemize}
The state of each replica is a set of edges, and with each new $\Write$ and $\Merge$, the graph known by each replica grows. The query function $\Query(s)$ outputs the set of edges in the state $s$ that are part of a loop in the graph represented by $s$.

The coordination function $\CF$ here allows any clause with the same set of inputs as the total input. We will see later that this coordination function is confluent.

\end{example}

Confluence is a property of an implementation, where the coordination function of the implementation allows any distribution, reordering, and repeating of individual inputs within the total input. Such a coordination function would be easy to implement practically.

\begin{definition}[Confluence]\label{def:confluence} 
An implementation is confluent if 
$$
\forall x \in \X \text{ and } c \in \Clause \text{ where } \Inputset(c) = x, \text{ we have } c \in \CF(x).
$$
\end{definition}
This implies that for any two clauses with identical input sets, the resulting query outputs will be identical. The outcome is independent of both the ordering and distribution of the inputs, meaning that no coordination is required to limit the selection of clauses.

\begin{example}[Distributed Garbage Collection]\label{example:garbage-collection}
Let's look at how to implement distributed garbage collection in our computation model. Distributed garbage collection can also be recognized as a graph problem: each node in the graph represents an object, and each directed edge in the graph represents an object reference. Any object that is not reachable from the root node (the one with number $0$, for example) can be garbage collected. The edges and nodes may be distributed across replicas, as shown in the example in Figure~\ref{fig:distributed-garbage-collection}.

\begin{figure}[ht]
    \centering
    \includegraphics[width=0.7\columnwidth]{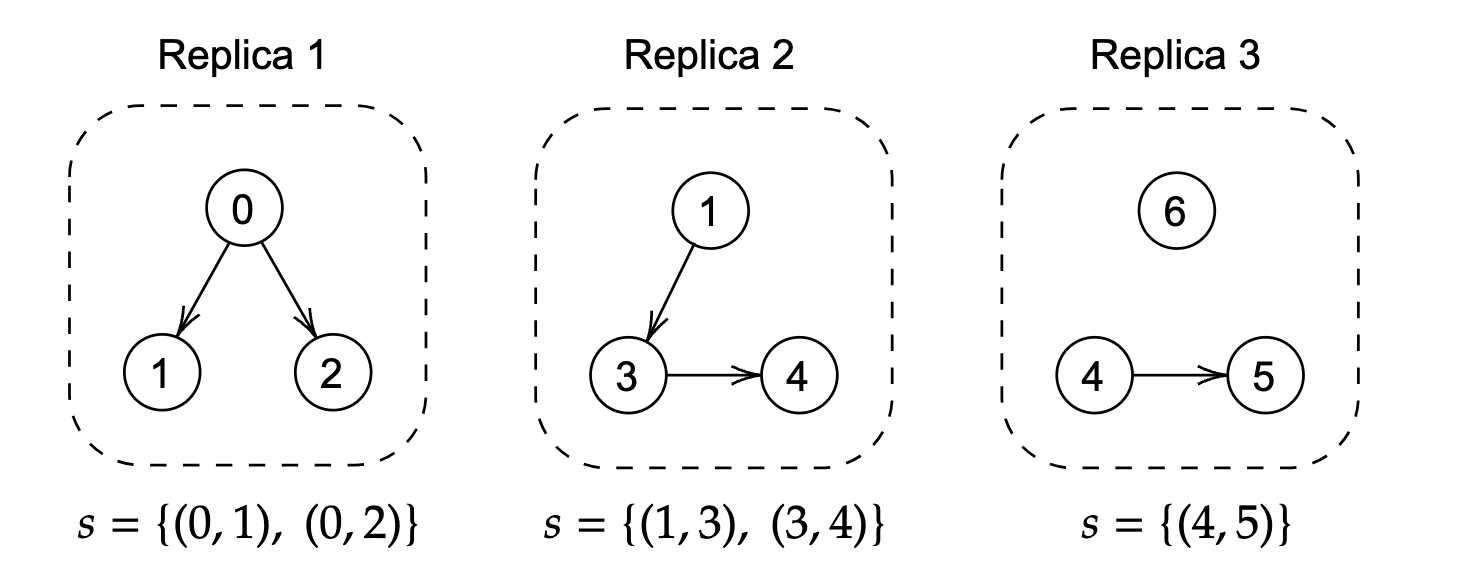}
    \caption{Distributed Garbage Collection Example}
    \label{fig:distributed-garbage-collection}
\end{figure}

An implementation of distributed garbage collection can have the same $\CF$, $s_0$, $\Write$, and $\Merge$ functions as the distributed deadlock problem in the implementation of Example \ref{example:deadlock3}. The only difference between the two problems is the  query function $\Query$: in distributed garbage collection, the $\Query$ function outputs the set of nodes that are not in the transitive closure of the root object.

The consistency partial order for this problem can also be represented by a subset relation $\subseteq$: $\forall v_1, v_2 \in \mathbb{V}, v_1 \leq v_2 \iff v_1 \subseteq v_2$. This consistency partial order definition has practical applications: it is safe to garbage collect all objects in set $v_1$ if for all future outputs $v_2$, $v_1 \subseteq v_2$. To be more concrete, if we can guarantee that the current set of objects that we want to garbage collect will never be referenced by the root node with new edges learned by the replica, then it is safe to garbage collect the current set of objects that are not referenced by the root node. 

However, we cannot give such guarantees in the problem of garbage collection if it is always possible to later get a new input that creates a path from the root node to a node previously marked as safe to collect. This leads to the root object referencing an object that has already been garbage collected and further leads to program faults. 
In order for the implementation to give a correct result, the coordination function has to guarantee that there will be no future inputs. 
We will later see that the fundamental reason is that the problem function $\Problem$ for garbage collection is not monotonic given the definition of consistency and the poset $(\mathbb{V}, \subseteq)$ .

\end{example}

In the example of distributed garbage collection, we see there is another form of coordination -- the guarantee of having received all inputs. For an implementation to be coordination-free, it has to not rely on this guarantee. Intuitively, an implementation needs to be correct even if it only has partial information about the inputs. We will formally define this as ``consistent under partition.''

First we need to formally define ``partition'' in our model of computation. This is especially intuitive in the form of clauses defined in definition~\ref{def:clause}. Every clause represents an execution trace, and a partition of the execution would be a ``sub-clause'' of the entire clause. For example, consider $c = c_1 \Merge c_2$. If $c$ is the entire clause, then $c_1$ and $c_2$ are both ``partitions'': $c_1$ represents the replica before receiving the remote state, and $c_2$ represents the remote replica. They are both partitions in the sense that they cannot guarantee that they have global information. If we take a step further, $c$ could be a partition for some other clause as well, because $c$ cannot guarantee that it has global information. Hence, ``partition'' is a relative notion -- a clause is a partition of some other clause. Specifically, this presents a partial order of clauses.

\begin{definition}[Partial Order of Clauses]
We define a partial order of clauses $\preceq$. $\forall c_1, c_2 \in \mathbb{C}$, $ c_1 \preceq c_2$ if any of the following is true:

\begin{itemize}
\itemsep=1ex
    \item $c_1 = c_2$,
    \item $\exists i \in \Inputs, c_2 = c_1 \Write i, $
    \item $\exists c \in \Clause, c_2 = c_1 \Merge c$
    \item $\exists c \in \Clause, c_2 = c \Merge c_1$
    \item $\exists c \in \Clause, c_1 \preceq c \mbox{ and } c \preceq c_2$
\end{itemize}
    
\end{definition}

If $c_1 \preceq c_2$, then $c_1$ is a partition of $c_2$. If the current execution trace on a replica is $c_1$, it could be a partition of some larger $c_2$; this also means that it could be a partition of another clause $c_3$, where $c_1 \preceq c_3$, or other clauses that include $c_1$. Any clause could be under partition because any clause cannot give the guarantee of having global information. This is where consistency under partition becomes important: the ability to guarantee consistency without the guarantee of global information.

\begin{definition}[Consistent Under Partition]\label{def:consistent-partition}
An implementation is \emph{consistent under partition} iff
$$
\begin{aligned}
& \forall x \in \X,  \forall c \in \CF(x), \\
& \forall c_0 \in \Clause \text{ such that } c_0 \preceq c, \Query(\Evaluate(c_0)) \leq \Query(\Evaluate(c))
\end{aligned}
$$
\end{definition}

An implementation that is consistent under partition is able to give consistent query results immediately without any coordination. Combined with confluence, such a coordination function is extremely simple to implement in practice: any replica can execute the write, merge, and query operations immediately upon request without consulting any other replica, while guaranteeing consistency at the same time. We call such implementations coordination-free:

\begin{definition}[Coordination-free]\label{def:coordination-free} 
An implementation is coordination-free if it is confluent and consistent under partition.
\end{definition}

\section{CALM Theorem}\label{sec:mono}

The CALM theorem (Consistency as Logical Monotonicity), proposed by Hellerstein and Alvaro~\cite{hellerstein_declarative_2010,hellerstein_keeping_2020}, proposes that a problem has a consistent, coordination-free distributed implementation if and only if it is monotonic. A proof of the theorem using transducer networks has been given by Ameloot, et al.~\cite{ameloot_relational_2010,ameloot_weaker_2016}. In this work, we give a more general proof of the CALM theorem from a distributed computation perspective rather than from a database point of view.

Hellerstein and Alvaro~\cite{hellerstein_keeping_2020} define consistency as confluence: ``...an operation on a single machine is confluent if it produces the same set of output responses for any non-deterministic ordering and batching of a set of input requests.'' But in the previous section, we see that confluence is actually a property implied by coordination-freeness. For an implementation to be coordination-free, it has to be confluent. For a coordination-free implementation, different replicas may receive events in different orders, and if the implementation is not confluent, different replicas will generate different results, violating the consistency that we desire. Thus we argue that consistency should not be defined as confluent, and confluence should be an aspect of coordination-freeness.

\begin{example}[Distributed Garbage Collection]\label{example:garbage-collection2}
In Example~\ref{example:garbage-collection}, we mentioned that the fundamental reason that the problem of distributed garbage collection does not have a coordination-free implementation is that it is not monotonic. In the problem of distributed garbage collection, the problem $\Problem\colon \X \to \Values$ has as its domain $\X$ and co-domain $\Values$ sets of edges, both endowed with a subset order to form a poset.
The problem definition is that $\Problem(x)$ is a subset of nodes that are not reachable in the graph from some root node, given the edges in $x$. Consider a graph with two nodes, $n_0$ and $n_1$, where $n_0$ is the root node. Then
$$
\Problem(\{(0,0), (1,1)\}) = \{ 1 \}
$$
and
$$
\Problem(\{(0,0), (1,1), (0,1)\}) = \emptyset .
$$
This function is not monotonic because $\{(0,0), (1,1)\} \subseteq \{(0,0), (1,1), (0,1)\}$, but $\{ 1 \} \nsubseteq \emptyset$.
\end{example}

Formally, we define monotonicity as the property of a function from a partially ordered set to another that preserves order. This aligns with the definition of monotonicity in order theory. To be more concrete, monotonicity is a property of the function $\Problem$ that preserves the order of input sets in the problem output. When more inputs are presented to the problem, the problem only yields outputs that are consistent with previous results.

\begin{definition}[Monotonic Problem]
A problem $(\Problem, \X, \Values, \leq)$ is monotonic iff
$$
\forall x_1, x_2 \in \X,
x_1 \subseteq x_2 \Rightarrow \Problem(x_1) \leq \Problem(x_2)
$$
\end{definition}

With this definition of monotonicity, we have carefully defined every aspect of the CALM theorem. We can now prove the theorem.

\begin{mytheorem}[Monotonicity $\Rightarrow$ Coordination-free]\label{theorem:monotonic-cf}
A monotonic problem has a consistent, coordination-free distributed implementation.
\begin{proof}
For any problem function $\Problem: \X \rightarrow \Values$, consider the following implementation:
\begin{enumerate}
\itemsep=0ex
    \item $\CF: \forall x \in \X, \CF(x)=\{c \in \Clause \mid \Inputset(c) = x\}$
    \item $\States = \X$, the set of possible states is same as the problem domain.
    \item $s_0=\emptyset$
    \item $\Write(s, i) = s \cup \{i\}$
    \item $\Merge(s_1, s_2) = s_1 \cup s_2$
    \item $\Query(s) = \Problem(s)$.
\end{enumerate}

(Step 1) We prove that this implementation is correct, meaning that it is a valid solution to the problem $\Problem$.

$\forall x \in \X, \forall c \in \CF(x), \Inputset(c)=x$ because of (1) above.

Since $\Write(s, i) = s \cup \{i\}$, $\Merge(s_1, s_2) = s_1 \cup s_2$, it is apparent that $\Evaluate_{\Write, \Merge}(c) = \Inputset(c)$. Then $\Query(\Evaluate(c)) = \Query(\Inputset(c)) = \Query(x) = \Problem(x)$.

(Step 2) We prove that this implementation is confluent.
Because in this implementation, $\CF(x)$ satisfies $\forall x \in \X, c \in \Clause, \Inputset(c) = x \Rightarrow c \in \CF(x)
$ in definition \ref{def:confluence}, this implementation is confluent. 

(Step 3) We prove that this implementation is confluent under partition. Because of (1) above,
$ \forall c \in \CF(x), \Inputset(c)=x $. 
$\forall c_0 \preceq c$, we define $x_0 = \Inputset(c_0)$, so, by definition, $x_0 = \Inputset(c_0) \subseteq \Inputset(c) = x$. Because of (1), $c_0 \in \CF(x_0)$. Since the implementation implements the problem $\Problem$, $\Query(\Evaluate(c_0)) = \Problem(x_0)$ and $\Query(\Evaluate(c)) = \Problem(x)$.
Because the problem is monotonic, $x_0 \subseteq x \Rightarrow \Problem(x_0) \leq \Problem(x)$. So $\Query(\Evaluate(c_0)) = \Problem(x_0) \leq \Problem(x) = \Query(\Evaluate(c))$, and, hence, the implementation is consistent under partition. 

(Step 4) Since the implementation is confluent and consistent under partition, this implementation is coordination-free.

\end{proof}
\end{mytheorem}

\begin{mytheorem}[Coordination-free $\Rightarrow$ Monotonicity]\label{theorem:cf-monotonic}
If a problem has a consistent, coordination-free distributed implementation, then the problem is monotonic.
\begin{proof}

Assume the coordination-free implementation for problem $\Problem$ is the coordination function $\CF$ and $\mathcal{RO} = (\Write, \Merge, \Query, s_0)$.

$\forall x_1, x_2 \in \X, x_1 \subseteq x_2$. Assume $x_1 = \{i_1, i_2, ...,i_n\}$ and $x_2=\{i_1, i_2, ...,i_n, i_{n+1}, ..., i_m\}$.

Since $\CF$ is confluent, 
$$
\begin{aligned}
& c_1=s_0 \Write i_1 \Write i_2 \Write ... \Write i_n \in \CF(x_1) \\
& c_2=s_0 \Write i_1 \Write i_2 \Write ... \Write i_n \Write i_{n+1} \Write ... \Write i_m \in \CF(x_2)
\end{aligned}
$$

With definition \ref{def:consistent-partition}, $c_1 \preceq c_2$.

Because the implementation is consistent under partition, according to definition~\ref{def:consistent-partition}, $\Query(\Evaluate(c_1)) \leq \Query(\Evaluate(c_2))$. Since this implementation implements problem $\Problem$, according to definition \ref{def:problem}, $\Problem(x_1) = \Query(\Evaluate(c_1)) \leq \Query(\Evaluate(c_2)) = \Problem(x_2)$. Thus the problem is monotonic.

\end{proof}
\end{mytheorem}

\begin{mycorollary}[CALM Theorem]
A problem has a consistent, coordination-free distributed implementation if and only if it is monotonic.
\begin{proof}
This follows directly from Theorems \ref{theorem:monotonic-cf} and \ref{theorem:cf-monotonic}.
\end{proof}
\end{mycorollary}

\section{Further Work: Availability}\label{sec:availability}

Our coordination-free computation model is local-first, meaning that each replica does not have to consult other nodes before beginning its computation. However, the model ignores timing. If you are willing to wait (possibly forever), it is often possible, in theory, to define a monotonic problem.
We begin with a mathematical corner case, which shows that the condition we really need mathematically is not monotonicity but rather the stronger condition of Scott continuity.

Consider Example~\ref{example:garbage-collection}, the garbage collection problem.
Suppose that instead of the $\Problem$ function in Example~\ref{example:garbage-collection2}, we define,
\begin{equation}\label{eq:continuity}
\Problem(x) = 
\begin{cases}
    \emptyset \mbox{, if not all inputs have been received} \\
    \mbox{set of nodes not reachable from the root node, otherwise.}
\end{cases}
\end{equation}

To be more concrete, we define each input to have a sequence number similar to what is shown in Section~\ref{sec:consistency}. As a mathematical corner case, in addition to the finite sets $D_n$, we can allow $\mathbb{N}$ in order to include infinite input sequences. For finite sequences, we add a special ``end token'' to the input alphabet $\Inputs$. This token denotes that this is the last edge in the graph, and any edges with larger sequence numbers can be ignored. With these additions, the problem (\ref{eq:continuity}) is well defined.

This function is trivially monotonic. The output $\Problem(x)$ needs to be the empty set until the input $x$ includes the end token \emph{and} all edges with a smaller sequence number.
If the inputs are not bounded, then this problem definition could result in waiting forever before yielding a result, but it is still mathematically well defined.

We can rule out the ``waiting forever'' case by further restricting problems to be Scott continuous rather than monotonic.
A function $\Problem\colon \X \to \Values$ is Scott continuous if given a directed subset $D \subset \X$,
$$
\Problem(\vee D) = \vee \Problem(D),
$$
where $\vee A$ is the least upper bound of a set $A$ and $\Problem(D)$ is the image of $D$ (the set $\{\Problem(d) ~|~ d \in D\}$).
It is easy to prove that any Scott continuous function is monotonic, but not all monotonic functions are Scott continuous (the function (\ref{eq:continuity}) above is a counterexample)~\cite{DaveyPriesly:90:Posets}.
If we restrict ourselves to finite sets, then every monotonic function is also continuous.
If we assume that, in practice, all program executions are finite, then this mathematical corner case is not very interesting.

Even for finite inputs, however, (\ref{eq:continuity}) does not describe a useful distributed garbage collection problem because monotonicity is achieved by waiting until all inputs are received before producing any useful result.
The CALM theorem, by itself, is agnostic to such timing considerations.
To consider timing, we need to augment the concept of a problem and consider interactivity.
Given such an augmentation, we should be able to derive useful relationships between the CALM theorem and the CAP~\cite{brewer2000cap,Brewer:17:CAP} and CAL\cite{LeeEtAl:23:CAL_CPS,LeeEtAl:23:CAL_IC,Brewer:17:CAP} theorems.
We leave this as a challenge problem for further work.

\section{Conclusions}\label{sec:conclusions}

In this paper, we present a model for distributed computation that separates coordination from computation. We formally defined problems and implementations, and introduced concepts in distributed computation such as strong eventual consistency, consistency, confluence, consistent under partition, and coordination-free. We gave two main theoretical results: necessary and sufficient conditions for strong eventual consistency, and a proof of the CALM theorem from a distributed computation perspective. 

We hope that our results can inspire deeper insight into coordination-free consistency for the distributed systems community, and make it easier for developers and researchers to argue which problems have a coordination-free implementation using the CALM theorem. This work also potentially opens up new research avenues in distributed computation about coordination and consistency by giving a model for computation that separates coordination and computation.

\printbibliography

\end{document}